\documentclass[10pt,twocolumn]{IEEEtran}
\usepackage{cite}
\usepackage[pdftex]{graphicx}
\usepackage[pdftex]{color}
\graphicspath{/}
\usepackage[ruled,vlined,linesnumbered]{algorithm2e}
\usepackage[cmex10]{amsmath}
\usepackage{algpseudocode}
\usepackage{array}
\usepackage{url}
\usepackage{bm} 
\usepackage{amssymb}

\usepackage{amsthm}
\newtheorem{thm}{Theorem}[section]

\newtheorem{lem}[thm]{Lemma}

\theoremstyle{remark}

\theoremstyle{definition}

\newcommand{\C}[1]{C\left( #1 \right)}
\newcommand{\SNIR}[2]{\frac{ #1 }{ 1 + #2}}

\begin{document}
\title{Low-Rate Machine-Type Communication via Wireless Device-to-Device (D2D) Links}

\author{\IEEEauthorblockN{Nuno K. Pratas and Petar Popovski}
\IEEEauthorblockA{\\Department of Electronic Systems, Aalborg University, Denmark \\
Email: \{nup,petarp\}@es.aau.dk}}
\maketitle

\begin{abstract}
Wireless cellular networks feature two emerging technological trends.  
The first is the direct Device-to-Device (D2D) communications, which enables direct links between the wireless devices that reutilize the cellular spectrum and radio interface.
The second is that of Machine-Type Communications (MTC), where the objective is to attach a large number of low-rate low-power devices, termed Machine-Type Devices (MTDs) to the cellular network.
MTDs pose new challenges to the cellular network, one if which is that the low transmission power can lead to outage problems for the cell-edge devices. Another issue imminent to MTC is the \emph{massive access} that can lead to overload of the radio interface. 
In this paper we explore the opportunity opened by D2D links for supporting MTDs, since it can be desirable to carry the MTC traffic not through direct links to a Base Station, but through a nearby relay. MTC is modeled as a fixed-rate traffic with an outage requirement. 
We propose two network-assisted D2D schemes that enable the cooperation between MTDs and standard cellular devices, thereby meeting the MTC outage requirements while maximizing the rate of the broadband services for the other devices.
The proposed schemes apply the principles Opportunistic Interference Cancellation and the Cognitive Radio's underlaying.
We show through analysis and numerical results the gains of the proposed schemes.
\end{abstract}



\section{Introduction}
\label{sec:Introduction}

Wireless cellular networks are embracing two technological trends, different from the mainstream technologies to improve data rates and coverage. The first trend is enabling of direct \emph{Device-to-Device (D2D)} communications, by enabling direct links between the wireless devices using the same spectrum and radio interface used for cellular communications. The second trend is that of \emph{Machine-Type Communications (MTC)} or \emph{Machine-to-Machine (M2M)} communications~\cite{boswarthick2012m2m,Estrin:1999:NCC:313451.313556}, where the objective is to attach a large number of low-rate low-power devices, termed \emph{Machine-Type Devices (MTDs)} to the cellular network. The drivers for this increase are the introduction of the smart-grid, large scale environment sensing, asset and health monitoring, among others.

The motivation for introducing D2D communication has primarily been seen in the possibility to offload the cellular infrastructure whenever identical content needs to be served in the same localized area, but also in the possibility to have reliable peer-to-peer links in a licensed spectrum for sharing e.g. multimedia content. In these cases the data volumes are relatively large and therefore is worthwhile to take advantage of the high-rate D2D links. 

We observe that it can be desirable to carry the machine-type traffic not through direct links to a Base Station (BS), but through a nearby relay. Therefore, in this paper we explore the opportunity opened by D2D links for supporting MTDs, for which the motivation is two-fold:
\begin{itemize}
\item \emph{Low power consumption:} MTD has the opportunity to connect to the infrastructure through a nearby device, then it can lower the transmission power;
\item \emph{Access load relief}: There could be potentially thousands of MTDs connected to the same cell~\cite{Shafiq:2012:FLC:2318857.2254767,5956587}, such that the management of many individual connections can incur a large overhead and inefficiency in the cellular system. If multiple MTDs are connected to a relaying node and the relaying node accesses the BS on behalf of those devices, then the connection management load on the BS can drastically decrease.
\end{itemize}
Critics may point out two counterarguments not to use relay for machine-type traffic.

The first is that the relay increases the connection latency; however, there is a large variety of machine-type applications and not all of them require delay in milliseconds~\cite{boswarthick2012m2m}. Examples include heat meters, health sensors, etc. and relaying can be restricted to the traffic that has a higher delay tolerance. 

The second counterargument is that it is not feasible to introduce infrastructure relay, and this sets the motivation to relay the MTC traffic through another, more capable device, such as smart phone, using the D2D link. The setup is shown on Fig.~\ref{fig:Scenario_Illustration}. The difference with the common usage of the D2D transmission is that here the D2D link is used for sending at a low, fixed rate and the objective is to attain certain outage probability. The scheme in this paper treats the D2D transmission from the MTD \textbf{M1} to \textbf{U1} as an \emph{underlay} transmission, occurring simultaneously with a downlink transmission from base station \textbf{B} that is intended for another device \textbf{U2}. This creates a multiple access channel at the receiver \textbf{U1}. However, the key observation is that both the power and the data rate of the MTD are low, which leads to high probability to decode the MTD signal. The decoding can occur either directly or after canceling the interference caused by the downlink signal from the BS. The latter is reminiscent of the concept of Opportunistic Interference Cancellation (OIC), proposed for cognitive radio~\cite{5714150}, where the primary signal can be canceled to increase the secondary data rate. This connection is due to the fact that a D2D link acts almost as a underlay cognitive radio. The subtle difference is that here the role of the primary transmission is taken by the BS downlink transmission, but the BS is aware about the D2D link and can therefore adjust the power/rate of its transmission in order to enhance the D2D transmission, unlike the standard, non-adaptive primary in a cognitive radio model.  
\begin{figure}
	\centering
		\includegraphics[width=\linewidth]{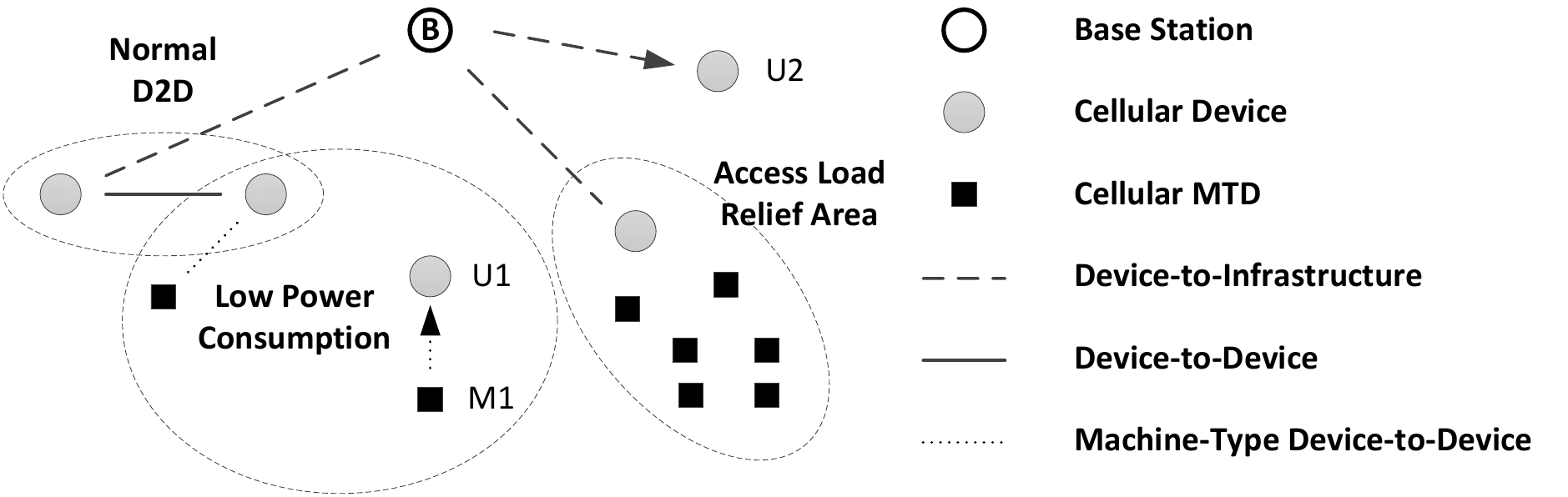}
	\caption{Scenario illustration.}
	\label{fig:Scenario_Illustration}
\end{figure}

All the observations above are the building blocks for the underlayed communication schemes proposed in this paper.
The paper is organized as follows.
In Section~\ref{sec:StateOfTheArt} we present the related work on D2D and MTC. While in Section~\ref{sec:SystemModel} we described the system model considered and the main assumptions.
In Section~\ref{sec:UnderlayedSchemesCharacterization} is given the analytical characterization of the proposed underlaying communication schemes, for which the performance is evaluated numerically in two different settings in Section~\ref{sec:NumericalResults}.
Finally, in Section~\ref{sec:Conclusion} we conclude the paper with a recap of the main conclusions and future outlook.

\section{Related Work}
\label{sec:StateOfTheArt}

The research literature and  3GPP feature several efforts to enable efficient MTC within GSM and LTE cellular networks~\cite{Chen2013,TR37.8682011}.
The current focus is mostly on the characterization of the effect of a massive amount of MTDs requesting access in a cell~\cite{5956587,Madueno2013,Shafiq:2012:FLC:2318857.2254767}, on adaptive contention and load control schemes to handle massive contention access~\cite{6477838,6336861,6093905,6211484} and on traffic aggregation in cellular devices from non-cellular networks~\cite{6507400}. In 3GPP there is an effort on enabling low cost LTE devices~\cite{6507400,3GPPTR36.888} for MTC, where the reduction in cost comes mainly from lowering the device bandwidth, baseband processing complexity, and using half-duplex transceivers.
From the many challenges to enable these devices to coexist with normal LTE devices~\cite{6215486}, the main one is the reduced coverage, which is commonly addressed by reducing the cell sizes, introducing relays, or increasing the transmission time.
The latter option is the one gaining traction~\cite{6215486,6507400}, since the former ones would lead to prohibitively expensive solutions.

In contrast, in this paper we motivate that the coverage of MTDs can be enhanced through the use of the D2D communication paradigm.
We note that within the literature dedicated to D2D, the synergies to be gained by combining MTC and D2D have not until now been identified, although the authors of~\cite{1673062} have motivated that the use of relaying in the uplink direction in cellular network is one of the key aspects to enable MTC within a cellular network.

Within the D2D context, there are already several proposals on underlaying~\cite{5350367,5208020,PekkaJANIS2009} and overlaying~\cite{5074679} D2D with a cellular network.
Underlaying schemes where the D2D shares the same air interface as the cellular network have been proposed first in~\cite{5350367,5208020,PekkaJANIS2009}.
Further, we note that the approaches currently applied to D2D are applicable to Vehicular-to-Vehicular (V2V), as shown in~\cite{6450121}. There are two main design directions currently found in the literature:
\begin{itemize}
	\item \emph{Network-assisted D2D}. The network performs all the decisions in regards to resource sharing mode selection (D2D or via the cellular infrastructure), power control, scheduling, selection of transmission format (such as modulation, coding rates, multi-antenna transmission mode, etc.)~\cite{5350367,5208020,5199353,5506248,5450284,PekkaJANIS2009,6163598,5073734,6450121,6162471,6047553,5910123,6364509,5645039};
	\item \emph{D2D with minimal network assistance}. The network provides at most only synchronization signals to the devices~\cite{5675775,6133537,5073611,6497679}.
\end{itemize}

In~\cite{5675650} the authors introduce three receiving modes for reliable D2D communication when the D2D UEs share the cellular radio resources. The first mode treats the interference as noise, the second mode decodes the interference and then cancels it. In low-interference regime, the first mode offers a higher sum rate, as in such regime treating interference as noise is optimal~\cite{5075903}. The second mode is applicable for the regime of very strong interference in which interference cancelation can be applied~\cite{4303349}. For the transient region between low and very strong interference, instead of the conventional approach with orthogonalization, the authors propose a mode in which the interference is retransmitted to the receiver which then cancels it from the original transmission. Nevertheless, the third mode is only feasible with multiple antennas\footnote{We note that the schemes proposed in this paper are not dependent on multiple antenna capability, although they could be greatly enhanced by it, but such discussion is left for future works.}.

The underlaying approach in this paper takes advantage of Opportunistic Interference Cancellation (OIC) proposed first in the context of cognitive radio in~\cite{5714150}. The main idea of OIC can be explained as follows. The receiver  $U$ observes a multiple access channel, created by the desired signal and an interfering signal. The interfering signal is a useful signal for a different node. However, if the current channel conditions and the rate selection allow, then $U$ can decode/cancel the interference and retrieve the desired signal at a higher rate. 

To the best of our knowledge, the MTC setting has not been considered within the context of D2D, underlaying and cognitive radio. 
The nearest match within cognitive radio can be found on~\cite{5291646} although there OIC was not considered.
We note that interference cancellation schemes have recently stirred increased attention both in academia~\cite{6182560,6155698} and industry~\cite{patent:6801580,patent:8433349}. 

The main contributions of the work presented in this paper, in contrast with the works currently present in the literature, can be summarized as follows:
\begin{itemize}
	\item Identification of the synergies between D2D, Cognitive Radio and MTC;
	\item Downlink (Uplink) underlaying scheme, which maximizes the downlink (uplink) cellular transmission rate, while attaining an outage rate and probability for the MTC D2D links, through the use of OIC at the cellular user (base station).
\end{itemize}

\section{System Model}
\label{sec:SystemModel}

Figure~\ref{fig:SystemModel_and_SchedulingOptions}(a) shows the network topology that captures the quintessential example of the transmission schemes proposed in this paper. The links will be denoted by X-Y, where it is understood that X is the transmitter and Y is the receiver. All nodes are half-duplex, i. e. a node cannot transmit and receive at the same time. 
These network nodes are characterized as follows:
\begin{itemize}
	\item A Base Station (B) capable of adapting the transmission rate and multipacket reception through Opportunistic Interference Cancellation (OIC)~\cite{5714150};
	\item Two \textit{normal} devices (U1 and U2) capable of adapting the transmission rate and multipacket reception through OIC; and
	\item Two MTDs M1, M2 with low complexity, unable to perform advanced processing, such as multipacket reception. They use low power and a low, fixed rate $R_M$.
\end{itemize}
\begin{figure}
	\centering
		\includegraphics[width=\linewidth]{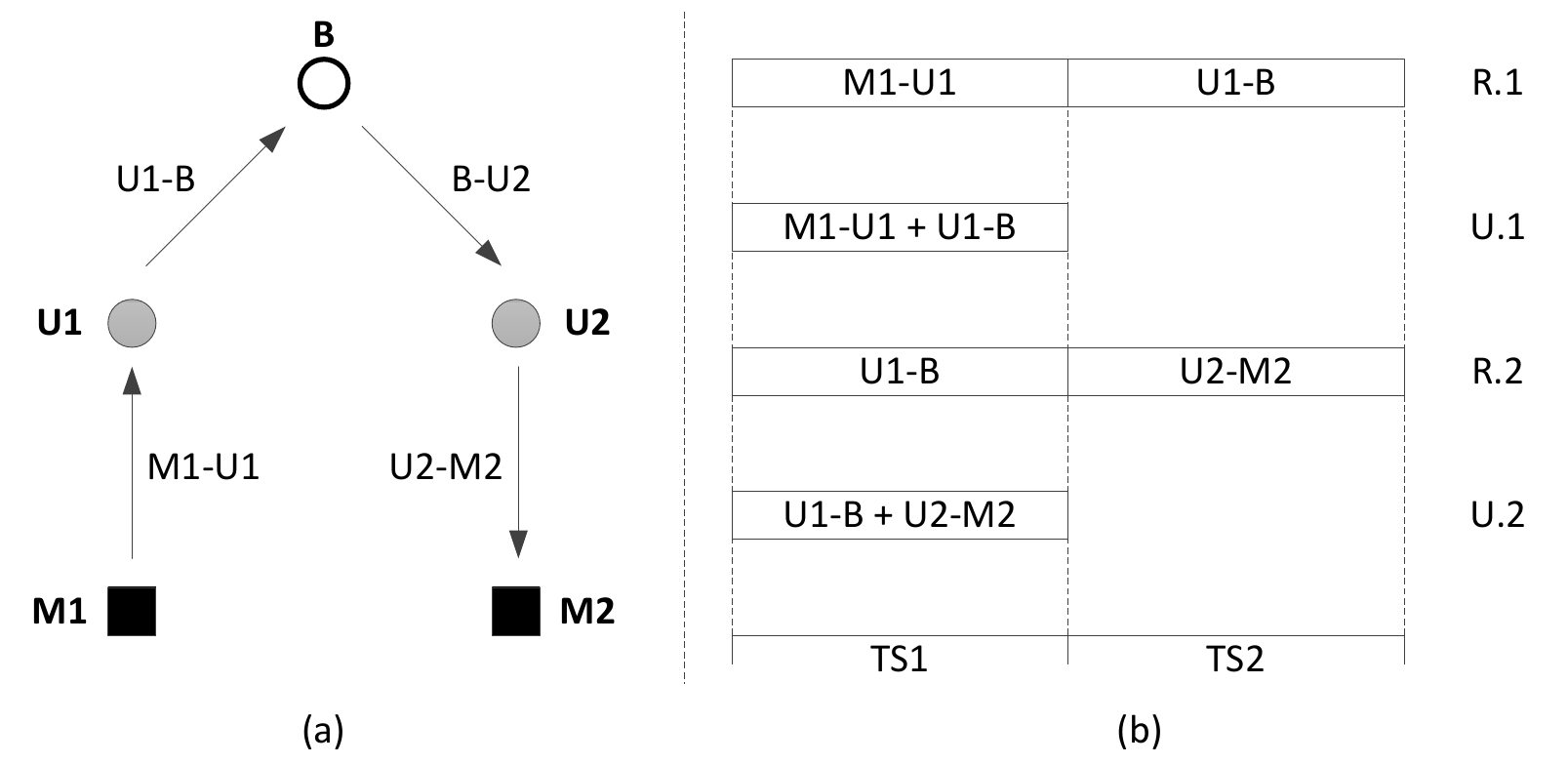}
	\caption{System model. (a) the four possible information flows; (b) The coupling and respective reference schemes duration in number of time slots.}
	\label{fig:SystemModel_and_SchedulingOptions}
\end{figure}

The main ingredient of the proposed scheme is the \emph{underlaying} i. e. simultaneous activation of the links M1-U1 and B-U2. Clearly, the underlaying would work regardless whether the transmission B-U1 carries data for U1 only or it also contains data that needs to be relayed over the link U2-M2 at a future point. Another underlaying occurs between the links U1-B and U2-M2. If the data sent from U2 is only relayed data that has been previously received from B, then the signal transmitted by U2 is completely known by B and can be canceled, such that B can receive the transmission from U1 over a ``clean channel''. Such schemes have been treated in~\cite{5722080,6384616,6310826,6259833,5962696,6216883}. However, here we leave the opportunity that U2 sends to M2 data that may have originated at U2 or at another MTD connected to U2, such that the signal sent by U2 cannot be assumed \emph{a priori} known by B. Therefore, B should try to decode the signal from U2 if it needs to cancel this interference. 

We propose a network-assisted underlaying solution that, based on the collected Channel State Information (CSI) and machine link outage requirements, selects the transmission power and maximal rate of the transmitters at cellular links U1-B and B-U2.
The links M1-U1 and U2-M2 are where the MTC occur. Any transmission to/from MTD is done over a D2D link and at a fixed rate $R_M$, with a requirement on the maximal outage probability of $P_{Out}$. On the other hand, the links U1-B and B-U2 adapt the rate/power to the channel conditions.

Fig.~\ref{fig:SystemModel_and_SchedulingOptions}(b) shows the proposed underlayed transmissions and the respective reference schemes, along with the required number of Time Slots (TS).
The total duration in time slots of each of these schemes is denoted as a \emph{scheduling epoch}. In the reference schemes the underlaying of a link with another does not take place.
All the transmissions have normalized bandwidth of $1$ Hz, therefore the duration of a time slot can be measured in terms of the number of symbols $N$.
The analysis here is done by considering that all communications are performed using capacity-achieving Gaussian codebooks. Furthermore, at most two simultaneous transmissions are considered over the multiple access channel, the other transmissions are modeled as noise. The receivers that apply OIC are aware of the codebooks used by the transmitters. We note that for MTC, the assumptions of sufficient large $N$ and Gaussian codebooks may not hold, since the MTC messages can be very short  and therefore finite block length theory should be used instead in their characterization~\cite{5452208}. In that sense, the results shown in this paper are optimistic, but the approach is valid and the follow-up work can rely on the finite block length theory. 

The underlaying (U.x) and respective reference schemes (R.x) are differentiated as follows:
\begin{itemize}
	\item (R.1): Due to the absence of underlaying, outage in the link M1-U1 occur only due to deep fades that render $R_M$ undecodeable;
	\item (U.1): The link M1-U1 underlays B-U2. An outage can occur both due to a deep fade on M1-U1 as well as due to the interference from B's transmission to U2. Device U1 can use OIC to cancel out the interference from B-U2 and thus decode M1-U1;
	\item (R.2) - similar to (R.1), outage occurs only due to a deep fade on U2-M2 and undecodeability of $R_M$;
	\item (U.2) - the link U2-M2 underlays the link U1-B. An outage can occur both due to deep fade on U2-M2 as well as due to the interference from the transmission from U1 to B. The Base Station B can use OIC to cancel out the interference from U2-M2 and afterward decode U1-B;
\end{itemize}

The rate/power applied on the links U1-B and B-U2 are decided at the network infrastructure side; therefore the term \emph{network-assisted}.
In this particular example, we consider that the decision is made at B and communicated to the network nodes using the appropriate signaling, but the details of it are outside the scope of this work.
Therefore B needs to know the CSI of possibly all communication links. 
This information can be available in the form of instantaneous CSI or statistical/average CSI.
Here we consider two different cases regarding the CSI availability. Fig.~\ref{fig:CSISchemes}(a) denotes full channel state information (F-CSI), where B is aware of the instantaneous channel realizations for all links.
Fig.~\ref{fig:CSISchemes}(b) denotes partial CSI (P-CSI), where B is aware of the instantaneous realization of the channel for the links between B and U1 and U2, while for each of the links between Mi-Ui, Mi-Uj, B it knows the average SNRs\footnote{We note that the implementation of a signaling procedure that enables F-CSI is expensive in terms of the required signaling overhead. Nevertheless, we do consider it here as it provides a stepping stone to obtain the P-CSI underlaying procedure.}.
\begin{figure}
	\centering
		\includegraphics[width=\linewidth]{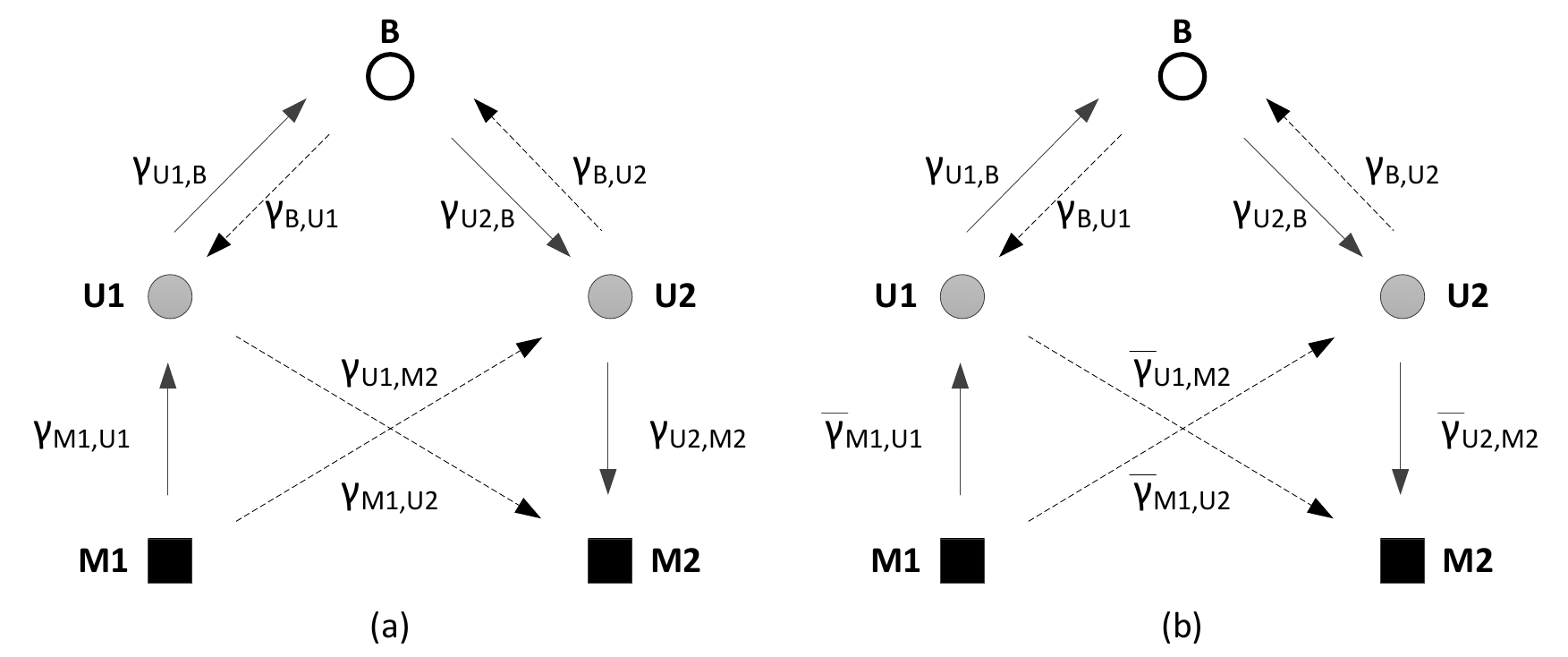}
	\caption{Types of CSI available: (a) Full-CSI and (b) Partial-CSI. The red lines denote interference links and the black ones signal links.}
	\label{fig:CSISchemes}
\end{figure}

All communication links are assumed to be non-Line-of-Sight and therefore the channel can be characterized by a block Rayleigh fading\footnote{In~\cite{5666195} the authors observed that with communications within the body area proximity, the short term fading follows a Rayleigh distribution. Another work \cite{Yazdandoost2009,6227423}, states that the more general Weibull distribution should be used instead.}.
Each of the complex channels $h_{X,Y}$ is reciprocal and Rayleigh-faded. The Signal-to-Noise Ratio (SNR) $\gamma_{X,Y}$, exponentially distributed, where the index $X$ denotes the signal transmitter, and $Y$ the receiver.
We assume for ease of analysis the block fading model, therefore the channel gains remain constant along the scheduling epoch.

\section{Characterization of the Underlayed Schemes}
\label{sec:UnderlayedSchemesCharacterization}

This section contains analysis of the two proposed underlayed schemes as well as formulation of the algorithms that adapt the rate and transmission power in the cellular links.
We first provide the preliminary analysis on the outage probability for MTC and this is followed by the analysis of the underlayed schemes (U.1) and (U.2). 
%

In Table~\ref{tab:ListOfSymbols} are listed the symbols used across this paper.
\begin{table}[t]
	\centering
			\begin{tabular}{ l  l }
		  	\hline
  			Symbol & Description\\
				\hline
				$\Gamma_x$ & Normalized Rate on direction $x$ \\
				$R_M$ & Fixed Rate for the machine-type link\\
				$R_{X,Y}$ & Rate of the link between X and Y\\
				$R_{X,Y}^\ast$ & Historical mean rate of the link between X and Y\\
				$S_M$ & Outage Margin \\
				$\gamma_{X,Y}$ & SNR from the signal from node X at receiver Y\\
				$P_X$ & Transmission power of node X\\
				$S_M$ & Outage Margin factor for the Machine-Type link\\
				$\sigma_Y^2$ & variance of the complex-valued Gaussian noise at receiver Y\\
				$h_{X,Y}$ & complex channel gain of the channel between the nodes X and Y\\
				
  			\hline
			\end{tabular}
	\caption{Definition of used symbols.}
	\label{tab:ListOfSymbols}
\end{table}

\subsection{Preliminaries}
\label{sec:Preliminaries}

Consider a time slot where only one transmission occurs, e.g. the first time slot of $R.1$ where only the communication link Mi-Ui is present.
Let, in general, X denote the transmitter and Y the receiver.
The instantaneous received signal at node $Y$, $y_{Y}$, is then given by
\begin{equation}
	y_{Y} = h_{X,Y} x_{X} + Z
\end{equation}
where $h_{X,Y}$ denotes the complex channel gain on the channel between the nodes $X$ and $Y$.
$x_{X}$ is the normalized signal transmitted by the $i$, with $E[|x_{X}|^2]=P_{X}$.
$Z$ is the complex-valued Gaussian noise with variance $\sigma_Y^2$.

The instantaneous SNR, $\gamma_{X,Y}$, is then given by,
\begin{equation}
	\gamma_{X,Y} = \frac{\left|h_{X,Y}\right|^2 x_{X}}{\sigma_Y^2}
\end{equation}
and the instantaneous achievable rate in this link is,
\begin{equation}
	R_{X,Y} = C(\gamma_{X,Y}).
	\label{rate}
\end{equation}

As stated in Section~\ref{sec:SystemModel}, the Machine-Type link has a fixed rate, hereinafter denoted as $R_M$.
Whenever the achievable instantaneous rate goes below the target rate, i.e. $R_{X,Y}<R_M$, the link is in outage.
From (\ref{rate}) the minimum SNR required to achieve $R_M$ is then given by $\gamma_{M}$.

The goal of the proposed schemes, is to perform underlaying whenever $R_{X,Y}>R_M$.
In the case of F-CSI, since B is aware of all instantaneous channel realizations, then the underlay communications can occur in a opportunistic way.
In this case the outage probability in the machine type links is upper bounded by the prescribed value. 
On the other hand, in the case of P-CSI, the underlaying can only take place if there is an outage margin factor, $S_M$, accounted for in the dimensioning of the Machine-Type link.
The following proposition can be stated:
\begin{lem}\label{lem:OutageMargin}
	The link X-Y attains an outage rate $R_M^{Out}$ and probability $P_{Out}$, in presence of $N$ interferers iff
	\begin{equation}
		P_X \geq S_M \cdot \frac{\sigma_Y^2}{\left|h_{X,Y}\right|^2} \cdot \frac{2^{\frac{R_M^{Out}}{1 - P_{Out}}}-1}{\ln \left( 1 - P_{Out} \right)}
	\end{equation}
	when $S_M = 1 + \frac{1}{\sigma_Y^2} \sum_k^N P_k \left| h_{k,Y} \right|^2$
\end{lem}
\begin{proof}
	The proof is found on Appendix~\ref{sec:ProofOfLemmaRef}.
\end{proof}

This paves the way for the underlayed schemes that we characterize in the following.

\subsection{Characterization of (U.1)}
\label{sec:CharacterizationOfIIii}

Here we characterize the underlayed scheme (U.1).
Consider the instantaneous received signal at Ui, $y_{Ui}$, which is given by,
\begin{equation}
	y_{Ui} = h_{Mi,Ui} x_{Mi} + h_{B,Ui} x_{B} + Z,
	\label{y_u1_s.2}
\end{equation}
and the instantaneous received signal, $y_{Uj}$, at Uj is now given by,
\begin{equation}
	y_{Uj} = h_{Mi,Uj} x_{M1} + h_{B,Uj} x_{B} + Z,
	\label{y_u2_s.2}
\end{equation}
where $h_{B,Ui}$, $h_{B,Uj}$, $h_{Mi,Uj}$ and $h_{Mi,Ui}$ denote the complex channel gains on the channel between the respective nodes identification in the indices.
$x_{B}$ and $ x_{Mi}$ are respectively the signal transmitted by the $B$ and $Mi$, such that $E[|x_{B}|^2] \leq P_B$ and $E[|x_{Mi}|^2] = P_{M}$.

From (\ref{y_u1_s.2}) and (\ref{y_u2_s.2}) follows that if the transmissions of both B and Mi are at the receiver, then the information-theoretic setting of the Gaussian Multiple Access Channel~\cite{gamal2011network} can be applied.
Since the distance from B to Ui and Uj is not necessarily equal and the same can be said about the distance from Mi to Ui and Uj, then synchronization of both B and Mi transmissions cannot be assured at both Ui and Uj.
Therefore, the multiple-access channel setting and consequentially OIC can only be applied at one of the receivers, which here we assume to be Ui in (U.1).
The required synchronization at Ui, can be accomplished by either B or Ui giving an appropriate timing advance to Mi.
While for the underlaying scheme (U.2), the synchronization occurs instead at B, since Mi cannot perform OIC.
Figure~\ref{fig:RateRegion_i_iii} depicts the signal and interference links of the underlaying scheme (U.1) and the multiple-access channel rate region at Ui.
\begin{figure}
	\centering
		\includegraphics[width=\linewidth]{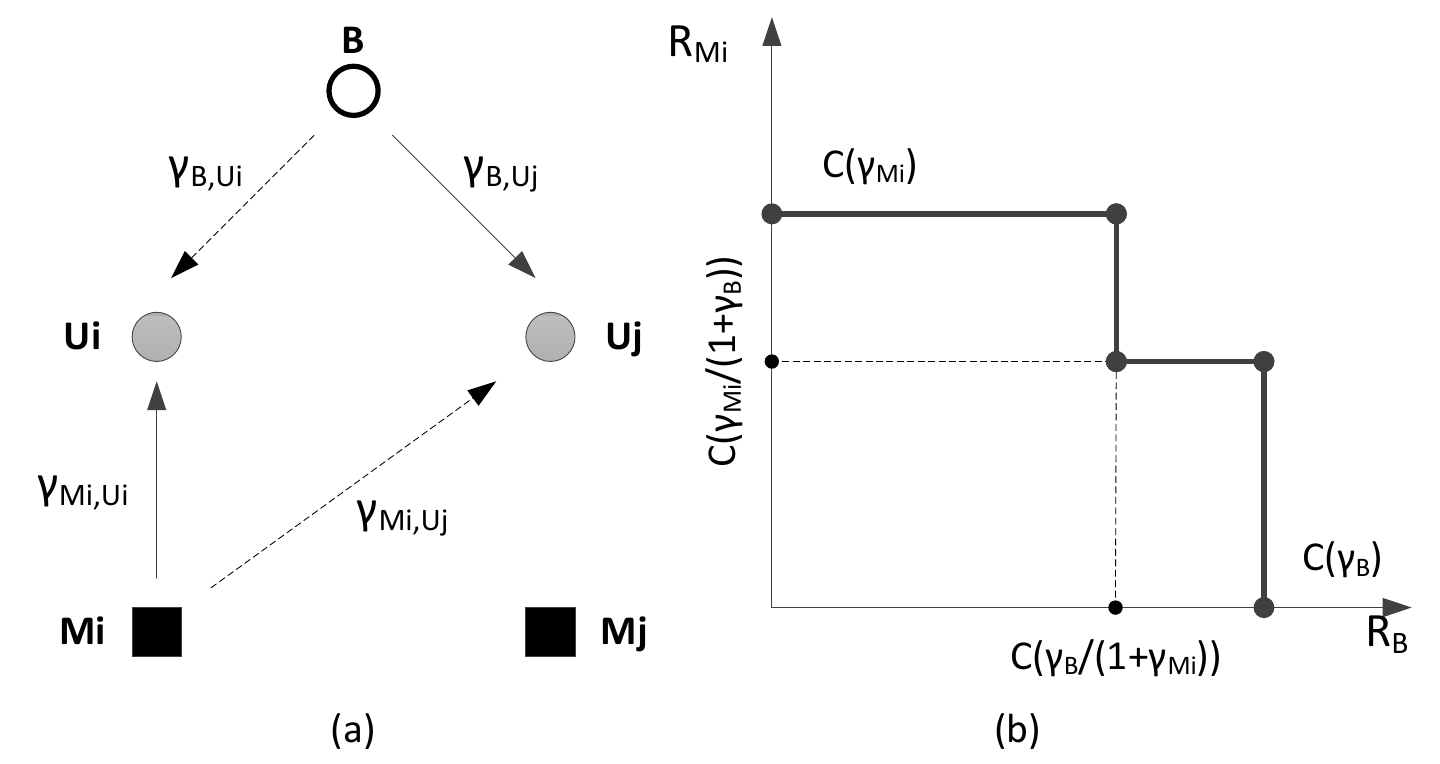}
	\caption{(a)Interference Links; (b) Rate region at Ui.}
	\label{fig:RateRegion_i_iii}
\end{figure}

At Uj the intended signal is the one coming from B, as shown in Figure~\ref{fig:RateRegion_i_iii}(a), while the one coming from Mi is treated as interference.
The maximum achievable instantaneous rate by B that is decodable at Uj, $R_{B,Uj}$, is then given by,
\begin{equation}
	R_{B,Uj} = \C{\frac{\gamma_{B,Uj}}{1 + \gamma_{Mi,Uj}}} \approx \C{\gamma_{B,Uj}}
\end{equation}
where this approximation can be made with negligible error since $P_B >> P_M$ and $h_{Mi,Uj}\approx h_{B,Uj}$, which leads to $\gamma_{B,Uj}>>\gamma_{Mi,Uj}$\footnote{We note that strictly, $R_{B,Uj}$ should be computed considering the interference contributed by Mi, and in the case of P-CSI this would lead to the presence of outage and therefore this rate should be computed following (\ref{outageRate}).}.

Figure~\ref{fig:RateRegion_i_iii}(b) shows the multiple-access channel rate region at the Ui receiver.
This capacity region is not convex~\cite{5714150}, as neither the B or Mi are able to do time-sharing, superposition coding or any other orthogonal resource sharing scheme, while the receiver does not use joint decoding and can only decode/cancel single-user codebooks. 
The shape of this capacity region is then delimited in regards to $R_{Mi}$ as follows:
\begin{itemize}
	\item If $0 < R_{Mi} \leq \C{\SNIR{\gamma_{Mi,Ui}}{\gamma_{B,Ui}}}$ - the signal from the Mi can be decoded even in the presence of interference from the signal transmitted from B, therefore, there is not limit set on $R_B$;
	\item If $\C{\SNIR{\gamma_{Mi,Ui}}{\gamma_{B,Ui}}} < R_{Mi} \leq \C{\gamma_{Mi,Ui}}$ - the signal from Mi can only be decoded if $R_B \leq \C{\SNIR{\gamma_{B,Ui}}{\gamma_{Mi,Ui}}}$.
	\item If $R_{M1} > \C{\gamma_{Mi,Ui}}$ - the signal from Mi cannot be decoded.
\end{itemize}

As stated previously, the goal of the proposed underlaying schemes is to maximize the rate of the cellular links while meeting the outage requirements of the machine type link.
This rate maximization procedure is dependent on the type of CSI available, therefore in the following we characterize this procedure for F-CSI and P-CSI.

\subsubsection{F-CSI}
\label{sec:FCSI_i_iii}

In the case of F-CSI, B, is aware of all instantaneous channel realizations at the network, as depicted in Figure~\ref{fig:CSISchemes}.
Therefore, B can always select the appropriate $R_B$ and $P_B$ so that the conditions given in Lemma~\ref{lem:OutageMargin} are respected.
Furthermore, we note that since Ui can perform OIC, then even in the case where there is no outage margin ($S_M = 1$), B can select the right combination of $R_B$ and $P_B$, so that the signal from B is decoded and cancelled, allowing the signal from Mi to be correctly decoded.
Therefore the outage probability in the link Mi-Ui, $P_{Out}$, is the same for which the Mi-Ui link was dimensioned for, and given as,
\begin{align}
	P\left[\C{\SNIR{\gamma_{Mi,Ui}}{\gamma_{B,Mi}}} < R_{M}\right] \leq P\left[\gamma_{Mi,Ui} < \gamma_{M}\right] = P_{Out}.
	\label{eq:PoutRequirement}
\end{align}
We note, that outage in the link Mi-Ui can still occur due to the Mi-Ui link conditions, but not due to interference from the B transmission.
In the case of P-CSI, B is not aware of the instantaneous channel realization of the Mi-Ui link, therefore the selection of the pair $R_B$ and $P_B$ is done instead based on the channel statistics, as shown in sub-section~\ref{sec:PCSI_i_iii}. 

Therefore, the maximum achievable instantaneous rate Mi, $R_{Mi}$, decodable at Ui, is given by,
\begin{equation}
		R_{Mi} \leq \begin{cases}
        \C{\gamma_{Mi,Ui}} 											 & R_{B,Ui} \leq \C{\SNIR{\gamma_{B,Ui}}{\gamma_{Mi,Ui}}}\\
        \C{\SNIR{\gamma_{Mi,Ui}}{\gamma_{B,Ui}}} & R_{B,Ui} >    \C{\SNIR{\gamma_{B,Ui}}{\gamma_{Mi,Ui}}}
		\end{cases},
		\label{eq:RM_i_iii}
\end{equation}
which leads to Lemma~\ref{lem:LimitsonRB}.
Where it is formalized the conditions that limit the maximum rate achievable by B, $R_B$.
\begin{lem}\label{lem:LimitsonRB}
	The maximum rate achievable by B, $R_B$, is limited by the channel conditions at Ui and $R_M$ iff
	\begin{equation}
		P_B \geq \left( \frac{P_M \left| h_{Mi,Ui} \right|^2}{2^{R_M}-1} - \sigma_U^2\right) \cdot \frac{1}{\left|h_{B,Ui}\right|^2}
		\label{eq:pb_limit_conditions}
	\end{equation}
\end{lem}
\begin{proof}
	Taking the second case of (\ref{eq:RM_i_iii}), we can write
	\begin{equation}
		R_M \geq \C{\SNIR{\gamma_{Mi,Ui}}{\gamma_{B,Ui}}} = \log_2\left(1 + \frac{P_M \left|h_{Mi,Ui}\right|^2}{\sigma_U^2 + P_B \left| h_{B,Ui} \right|^2}\right)
	\end{equation}
	which can be manipulated to held (\ref{eq:pb_limit_conditions}), proving the Lemma~\ref{lem:LimitsonRB} claim.
	We note that when (\ref{eq:pb_limit_conditions}) does not hold, then $R_B$ is not limited by the channel conditions at Ui and neither by $R_M$.
\end{proof}

The maximum achievable instantaneous rate for B, $R_B$, according with the shape of multiple access channel capacity region is given by (\ref{eq:RB_i_iii}).

\begin{figure*}[!t] 
\normalsize
\begin{equation}
\label{eq:RB_i_iii}
R_{B} \leq \\
		\begin{cases}
		0 & R_M > \C{\gamma_{Mi,Ui}} \\ 
		
		\min \left[\C{\SNIR{\gamma_{B,Ui}}{\gamma_{Mi,Ui}}}, \C{\gamma_{B,Uj}}\right] & \C{\SNIR{\gamma_{Mi,Ui}}{\gamma_{B,Ui}}} \leq R_M \leq \C{\gamma_{Mi,Ui}} \\
		
		\C{\gamma_{B,Uj}} & R_M < \C{\SNIR{\gamma_{Mi,Ui}}{\gamma_{B,Ui}}} 
		
		\end{cases}\end{equation}
\hrulefill
\vspace*{4pt}
\end{figure*}

The first case of (\ref{eq:RB_i_iii}) is set by design, since in a network composed by multiple Ui-Mi pairs, this particular underlaying would not occur and instead an alternative Ui-Mi pair would be chosen.

From the second and third case of (\ref{eq:RB_i_iii}) and assuming that B has a transmission power upper bound given by $P_B^{Max}$, then Theorem~\ref{thm:MaximumRB} can be stated as follows.
\begin{thm}\label{thm:MaximumRB}
	The maximum rate that $B$ can take, $R_B$, is given by
	\begin{equation}
		R_B = \C{\gamma_{B,Uj}^{(1)}}
		\label{ep:RB_first_claim}
	\end{equation}
	when
	\begin{equation}
		P_B^{Max} < \left( \frac{P_M \left| h_{Mi,Ui} \right|}{2^{R_M}-1} - \sigma_U^2\right) \cdot \frac{1}{\left|h_{B,Ui}\right|^2}
		\label{ep:PB_first_claim}
	\end{equation}
	otherwise
	\begin{equation}
		R_B = \max\left[ \C{\gamma_{B,Uj}^{(2)}}, \min\left( \C{\frac{\gamma_{B,Ui}^{(1)}}{1 + \gamma_{Mi,Ui}}}, \C{\gamma_{B,Ui}^{(1)}} \right) \right]
		\label{ep:RB_second_claim}
	\end{equation}
	where
	\begin{equation}
		\gamma_{B,Ui}^{(1)} \leftarrow P_B^{(1)} = P_B^{Max}
	\end{equation}
	and
	\begin{equation}
		\gamma_{B,Ui}^{(2)} \leftarrow P_B^{(2)} = \left( \frac{P_M \left| h_{Mi,Ui} \right|}{2^{R_M}-1} - \sigma_U^2\right) \cdot \frac{1}{\left|h_{B,Ui}\right|^2}
	\end{equation}
\end{thm}

\begin{proof}
	
	The claims given by (\ref{ep:RB_first_claim}) and (\ref{ep:PB_first_claim}) are proved directly using Lemma~\ref{lem:LimitsonRB}, as it follows that if,
	\begin{equation}
		P_B^{Max} < \left( \frac{P_M \left| h_{Mi,Ui} \right|}{2^{R_M}-1} - \sigma_U^2\right) \cdot \frac{1}{\left|h_{B,Ui}\right|^2}
	\end{equation}
	then $R_B$ is not dependent on the channel conditions, then $R_B$ is maximized if $P_B = P_B^{Max}$ and $R_B$ is given by (\ref{ep:RB_first_claim}).
	
	The claim given by (\ref{ep:RB_second_claim}) can be proved by considering that (\ref{eq:RB_i_iii}) can be manipulated, using the result from Lemma~\ref{lem:LimitsonRB}, to held,
	\begin{equation}
			R_{B} \leq \begin{cases}
		
			\min \left[\C{\SNIR{\gamma_{B,Ui}}{\gamma_{Mi,Ui}}}, \C{\gamma_{B,Uj}}\right] & P_B \geq P_B^{\ast} \\
			\C{\gamma_{B,Uj}} & P_B < P_B^{\ast}
			
			\end{cases},
			\label{eq:RB_i_iii_proof}
	\end{equation}
	from which the pair $P_B^{(m)}$ and $R_B^{(m)}$ is selected such that $R_B$ is maximized and where $P_B^{\ast}$, is given by,
	\begin{equation}
		P_B^{\ast} = \left( \frac{P_M \left|h_{Mi,Ui}\right|^2}{2^{R_M}-1} - \sigma_U^2 \right) \frac{1}{\left| h_{B,Ui} \right|^2}
	\end{equation}
\end{proof}

The procedure to obtain then maximum rate, $R_B$, is formalized in Algorithm~\ref{alg:MaximumRBComputationFCSI} according with the guidelines provided by Theorem~\ref{thm:MaximumRB}.
\begin{algorithm}[t]
 \DontPrintSemicolon
 \SetAlgoLined
 \LinesNumbered

	\textbf{Input Parameters:}
	$P_B^{Max}$, $P_M$, $R_M$, $\left|h_{Mi,Ui}\right|^2$, $\left|h_{B,Ui}\right|^2$, $\left|h_{B,Uj}\right|^2$\;
	
	\textbf{Ouput Parameters:}
	$P_B$ and $R_B$\;
	
	\textbf{Algorithm:}\;
	$P_B^{(1)} = P_B^{Max} \rightarrow \gamma_{B,Ui}^{(1)}, \gamma_{B,Uj}^{(1)}$\;
	$R_B^{(1)} = \max \left[\C{\SNIR{\gamma_{B,Uj}^{(1)}}{\gamma_{Mi,Ui}}}, \C{\gamma_{B,Uj}^{(1)}}\right]$\;
	
	$P_B^{(2)} = \min\left[ \left(\frac{P_M\left|h_{Mi,Ui}\right|^2}{2^{R_M}-1}-\sigma_U^2\right)\frac{1}{\left|h_{B,Ui}\right|^2}, P_B^{Max}\right] \rightarrow \gamma_{B,Uj}^{(2)}$\;
	$R_B^{(2)} = \C{\gamma_{B,Uj}^{(2)}}$\;
	
	\eIf{$R_B^{(1)}\geq R_B^{(2)}$}{
		$P_B = P_B^{(1)}$\;
		$R_B = R_B^{(1)}$\;
	}{
		$P_B = P_B^{(2)}$\;
		$R_B = R_B^{(2)}$\;
	}
	
	\caption{Procedure to compute $P_B$ and $R_B$ in the case of F-CSI.}
\label{alg:MaximumRBComputationFCSI}
\end{algorithm}

\subsubsection{PCSI}
\label{sec:PCSI_i_iii}

In the case of P-CSI, as discussed in Section~\ref{sec:SystemModel}, B is aware of the instantaneous channel realizations in the cellular links, but is only aware of the mean of the channel realization $\left|h_{Mi,Ui}\right|^2$.
From Theorem~\ref{thm:MaximumRB} follows that $R_B$ and $P_B$ needs to be chosen such that the target outage probability, $P_{Out}$, in the link Mi-Ui is met, i.e. such that (\ref{eq:PoutRequirement}) is valid.

The outage probability given by the left term of (\ref{eq:PoutRequirement}) can only be evaluated numerically.
For this purpose the methodology exposed in~\cite{5714150} can be applied.
Therefore, the maximum $R_B$ and $P_B$ can only be found through an exaustive search\footnote{We note that more computational efficient search methods could be applied, such as bi-dimensional bisection search~\cite{Shellman2002641} among other random search methods.} of all possible $R_B$ and $P_B$ pairs, with some quantization error.
In Algorithm~\ref{alg:MaximumRBComputationPCSI}, is shown a sub-optimal search approach on the domain of $R_B$ based on the bisection method, where it is assumed that $P_B = P_{B}^{Max}$.
\begin{algorithm}[t]
 \DontPrintSemicolon
 \SetAlgoLined
 \LinesNumbered

	\textbf{Input Parameters:}
	$P_B^{Max}$, $\left|h_{B,Uj}\right|^2$, $\left|h_{B,Ui}\right|^2$, $E\left[\left|h_{Mi,Ui}\right|^2\right]$, $\epsilon$\;
	
	\textbf{Ouput Parameters:}
	$R_B$\;
	
	\textbf{Algorithm:}\;
	
  $R_B^{Min} = 0$\;
	$R_B^{Max} = 0$\;
	
	$R_B^{Current} = \frac{\left( R_B^{Min} + R_B^{Max} \right)}{2}$\;
	
	\While{$\left| P_{Out} - P_{Out}^{*}\left( R_B^{Current} \right) \right| > \epsilon$}{
		
			$R_B^{Current} = \frac{\left( R_B^{Min} + R_B^{Max} \right)}{2}$\;
			
			\eIf{$P_{Out}^{*}\left( R_B^{Current} \right) > P_{Out}$}{
				$R_B^{Max} = R_B^{Current}$\;
			}{	
				$R_B^{Min} = R_B^{Current}$\;			
			}
	}
	
	$R_B =  R_B^{Current}$\;
	
	\caption{Procedure to compute $R_B$ in the case of P-CSI, given an $P_{Out}$, and a stopping criterium $\epsilon$ and where $P_{Out}^{*}\left(R_{B}^{Current},\right)$ is computed following the method exposed in~\cite{5714150}.}
\label{alg:MaximumRBComputationPCSI}
\end{algorithm}

\subsection{Characterization of (U.2)}
\label{sec:CharacterizationOfIiIv}

In Figure~\ref{fig:RateRegion_ii_iv}(a) is depicted the signal and interference links in the case where the link Ui-Mi underlays the link Uj-B.
\begin{figure}
	\centering
		\includegraphics[width=\linewidth]{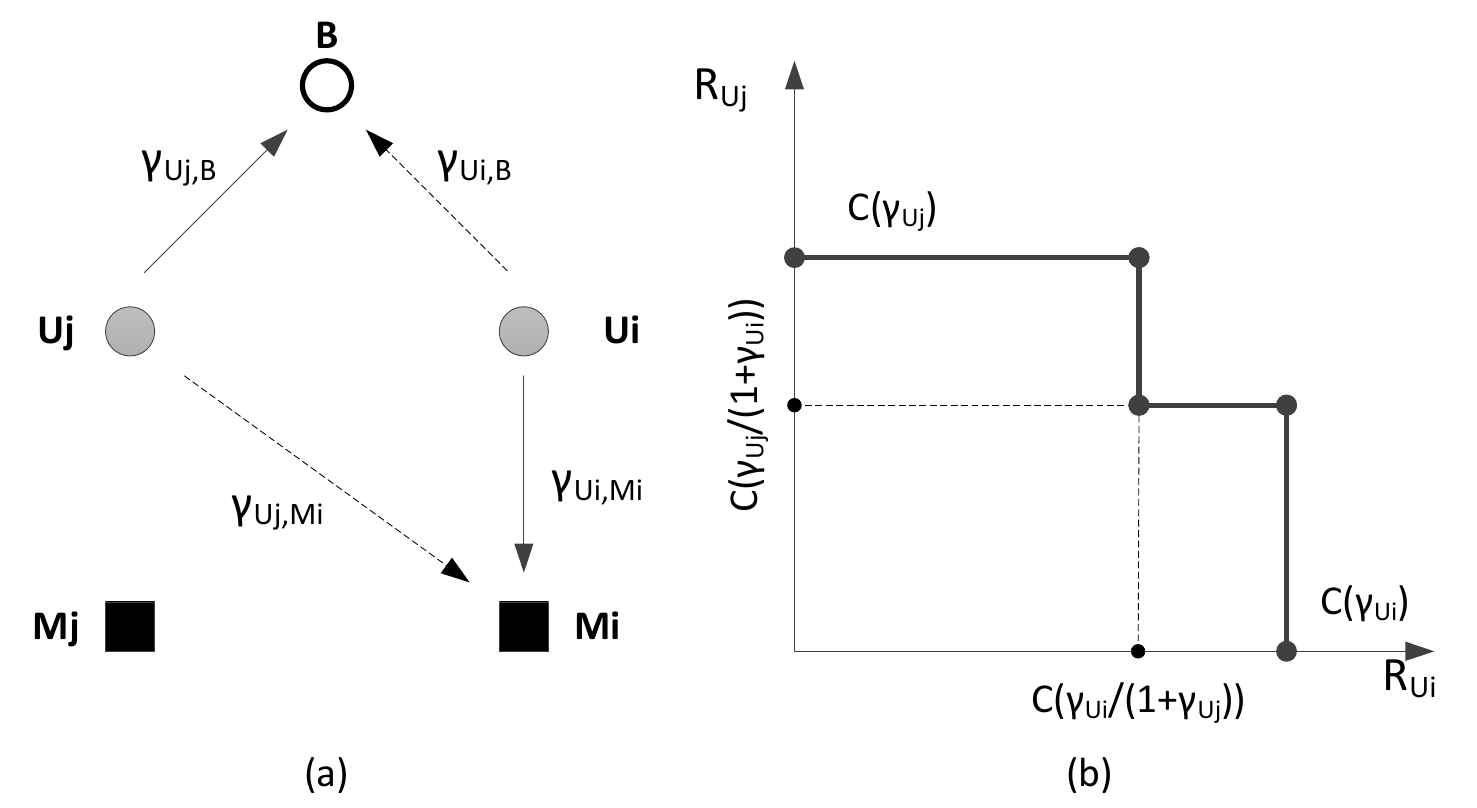}
	\caption{(a) Signal and Interference Links; (b) Rate Region at B.}
	\label{fig:RateRegion_ii_iv}
\end{figure}

The instantaneous received signal, $y_{B}$, at B is given by,
\begin{equation}
	y_{B} = h_{Uj,B} x_{Uj} + h_{Ui,B} x_{Ui} + Z,
	\label{y_b_s.3}
\end{equation}
and the instantaneous received signal, $y_{Mi}$, at Mi is given by,
\begin{equation}
	y_{Mi} = h_{Uj,Mi} x_{Uj} + h_{Ui,Mi} x_{U2} + Z,
	\label{y_m_s.3}
\end{equation}
where $h_{Uj,B}$, $h_{Ui,B}$, $h_{Uj,Mi}$ and $h_{Ui,Mi}$ denote the complex channel gain of the channels between the respective nodes identification in the indices.
$x_{Uj}$ and $x_{Ui}$ are respectively the signal transmitted by the Uj and Ui, such that $E[|x_{Uj}|^2] \leq P_U^{Max}$ and $E[|x_{Ui}|^2] \leq P_U^{Max}$.

As in the previous sub-section, in the following we characterize the underlaying scheme (U.2) in F-CSI and P-CSI conditions.

\subsubsection{F-CSI}
\label{sec:FCSI_ii_iv}

At Mi the signal from Ui is decoded in the presence of interference from Uj, since as stated in Section~\ref{sec:SystemModel} the MTDs are not able to perform OIC.
Therefore, $R_{Ui,Mi}$ is only affected by $P_{Uj}$ and the maximum instantaneous rate from Ui that can be decoded at Mi, $R_{Ui}$, is given by,
\begin{equation}
	R_{Ui,Mi} = \C{\SNIR{\gamma_{Ui,Mi}}{\gamma_{Uj,Mi}}}
	\label{eq:ru2m2}
\end{equation}
This leads to the following Lemma:
\begin{lem}\label{lem:MaximumPU1_MaxPU2}
		If $P_{Ui} = P_U^{Max}$ then the allowed $P_{Uj}$ is maximized.
\end{lem}
\begin{proof}
	From (\ref{eq:ru2m2}), when $R_{Ui,Mi} = R_M$ then,
	\begin{equation}
		R_M = R_{Ui,Mi} \leq \log_2\left( 1 + \frac{P_{Ui}\left| h_{Ui,Mi} \right|^2}{\sigma_M^2 + P_{Uj} \left| h_{Uj,Mi} \right|^2} \right)
	\end{equation}
	which after manipulation helds,
	\begin{equation}
		P_{Uj} \leq P_{Uj}^{\ast}
		\label{eq:PU1_M2_Upper_Bound}
	\end{equation}
	where $P_{Uj}^{\ast}$ is given by,
	\begin{equation}
		P_{Uj}^{\ast} = \left( \frac{P_{Ui} \left| h_{Ui,Mi} \right|^2}{2^{R_{Ui}} - 1} - \sigma_M^2 \right) \frac{1}{\left|h_{Uj,Mi}\right|^2}.
	\end{equation}
	Therefore, (\ref{eq:PU1_M2_Upper_Bound}) gives an upper bound on $P_{Uj}$, which is maximized when $P_{Ui} = P_U^{Max}$, while $P_{Uj} \leq  P_U^{Max}$, which is inline with the result from Lemma~\ref{lem:OutageMargin}.
\end{proof}

From the result from Lemma~\ref{lem:MaximumPU1_MaxPU2} is then assumed\footnote{The main consequence of this assumption is that the optimality of the proposed procedure is no longer assured.} that $P_{Ui} = P_U^{Max}$, so to increase the range where $P_{Uj}$ is valid. 

From (\ref{y_b_s.3}) follows that if the transmissions of both Uj and Ui are synchronized at B, then the information-theoretic setting of the Gaussian Multiple Access Channel~\cite{gamal2011network} can be applied.
In Figure~\ref{fig:RateRegion_ii_iv}(b) is shown the multiple access channel rate region at B. 

The maximum achievable instantaneous rate for Uj, $R_{Uj}$, according with the shape of the multiple access channel rate region depicted in Figure~\ref{fig:RateRegion_ii_iv}(b), is given by,
\begin{equation}
		R_{Uj} \leq \begin{cases}
	
		\C{\gamma_{Uj,B}} & R_M \leq \C{\SNIR{\gamma_{Ui,B}}{\gamma_{Uj,B}}}  \\
		
		\C{\SNIR{\gamma_{Uj,B}}{\gamma_{Ui,B}}} & R_M > \C{\SNIR{\gamma_{Ui,B}}{\gamma_{Uj,B}}}
				
		\end{cases}.
		\label{eq:U1_ii_iv}
\end{equation}
\begin{lem}\label{lem:PU1_UpperBound}
		When $P_{Uj}$ is upper bounded as follows,
		\begin{equation}
			P_{Uj} \leq \left( \frac{P_{Ui} \left| h_{Ui,B} \right|^2}{2^{R_M} - 1} - \sigma_B^2\right) \frac{1}{\left| h_{Uj,B}\right|^2}
			\label{eq:P_U1_UpperBoundAtB}
		\end{equation}
		then, the signal from Ui can be decoded and the achievable rate by Uj given by
		\begin{equation}
			R_{Uj} \leq \C{\gamma_{Uj,B}}
			\label{eq:R_U1_Interference_Cancelled}
		\end{equation}
\end{lem}
\begin{proof}
		The proof follows the same steps of the proof of Lemma~\ref{lem:LimitsonRB}.
		From the first case of (\ref{eq:U1_ii_iv}), (\ref{eq:R_U1_Interference_Cancelled}) occurs when,
		\begin{equation}
			R_M \leq \C{\SNIR{\gamma_{Ui,B}}{\gamma_{Uj,B}}} = \log_2\left(1 + \frac{P_{Ui} \left|h_{Ui,B}\right|^2}{\sigma_B^2 + P_{Uj} \left| h_{Uj,B} \right|^2}\right)
		\end{equation}
		which can be manipulated to held (\ref{eq:P_U1_UpperBoundAtB}).
\end{proof}

So, the following theorem can be stated in regards to the maximum value that $R_{Uj}$ can take:
\begin{thm}\label{thm:MaximumRU1}
	The upper bound on $R_{Uj}$ is given by,
	\begin{equation}
		R_{Uj,B} \leq \max \left[ \C{\gamma_{Uj,B}^{(1)}}, \C{\SNIR{\gamma_{Uj,B}^{(2)}}{\gamma_{Ui,B}}} \right]
	\end{equation}
	where,
	\begin{equation}
	P_{Uj}^{(1)} \leq \min\left[ P_{U}^{Max}, \left( \frac{P_{Ui} \left| h_{Ui,B} \right|^2}{2^{R_M} - 1} - \sigma_B^2\right) \frac{1}{\left| h_{Uj,B}\right|^2}, P_{Uj}^{\ast} \right]
	\end{equation}
	and
	\begin{equation}
		P_{Uj}^{(2)} = \min \left[P_{U}^{Max}, \left( \frac{P_{Ui} \left| h_{Ui,Mi} \right|^2}{2^{R_{Ui}} - 1} -\sigma_M^2 \right) \frac{1}{\left|h_{Uj,Mi}\right|^2} \right]
	\end{equation}

\end{thm}
\begin{proof}
	From Lemma~\ref{lem:PU1_UpperBound} follows that, $R_{Uj,B}$ can be upper-bounded as,
	\begin{equation}
			R_{Uj} \leq \begin{cases}
	
			\C{\gamma_{Uj,B}} & P_{Uj} \leq \left( \frac{P_{Ui} \left| h_{Ui,B} \right|^2}{2^{R_M} - 1} - \sigma_B^2\right) \frac{1}{\left| h_{Uj,B}\right|^2}  \\
		
			\C{\SNIR{\gamma_{Uj,B}}{\gamma_{Ui,B}}} & P_{Uj} > \left( \frac{P_{Ui} \left| h_{Ui,B} \right|^2}{2^{R_M} - 1} - \sigma_B^2\right) \frac{1}{\left| h_{Uj,B}\right|^2}
				
			\end{cases}.
			\label{eq:U1_ii_iv_PU1_Based}
	\end{equation}

	The first case from (\ref{eq:U1_ii_iv_PU1_Based}) leads that $P_{Uj}$ is upper bounded by (\ref{eq:P_U1_UpperBoundAtB}), but it is also upper bounded by (\ref{eq:PU1_M2_Upper_Bound}) due to Lemma~\ref{lem:MaximumPU1_MaxPU2} and finally by $P_{U}^{Max}$.
	
	The second case from (\ref{eq:U1_ii_iv_PU1_Based}) leads that $P_{Uj}$ is upper bounded by $P_{U}^{Max}$ and by (\ref{eq:PU1_M2_Upper_Bound}) due to Lemma~\ref{lem:MaximumPU1_MaxPU2}.
\end{proof}

The procedure to obtain then maximum rate, $R_{Uj}$, is formalized in Algorithm~\ref{alg:MaximumRU1ComputationFCSI} according with the guidelines provided by Theorem~\ref{thm:MaximumRU1}.
\begin{algorithm}[t]
 \DontPrintSemicolon
 \SetAlgoLined
 \LinesNumbered

	\textbf{Input Parameters:}
	$P_{Uj}^{Max}$, $R_M$, $\left|h_{Uj,B}\right|^2$, $\left|h_{Ui,B}\right|^2$, $\left|h_{Ui,M}\right|^2$, $\left|h_{Uj,M}\right|^2$\;
	
	\textbf{Ouput Parameters:}
	$P_{Uj}$ and $R_{Uj}$\;
	
	\textbf{Algorithm:}\;
	$P_{Uj}^{(1)} = \min\left[ P_{U}^{Max}, (\ref{eq:PU1_M2_Upper_Bound}), (\ref{eq:P_U1_UpperBoundAtB}) \right]$\;
	$R_{Uj}^{(1)} = \C{\gamma_{U1,B}^{(1)}}$\;
	
	$P_{Uj}^{(2)} = \min \left[P_{U}^{Max}, (\ref{eq:PU1_M2_Upper_Bound}) \right]$\;
	$R_{Uj}^{(2)} = \C{\SNIR{\gamma_{Uj,B}^{(2)}}{\gamma_{Ui,B}}}$\;
	
	\eIf{$R_{U1}^{(1)} \geq R_{Uj}^{(2)}$}{
		$P_{Uj} = P_{Uj}^{(1)}$\;
		$R_{Uj} = R_{Uj}^{(1)}$\;
	}{
		$P_{Uj} = P_{Uj}^{(2)}$\;
		$R_{Uj} = R_{Uj}^{(2)}$\;
	}
	
	\caption{Procedure to compute $P_{Uj}$ and $R_{Uj}$ in the case of F-CSI.}
\label{alg:MaximumRU1ComputationFCSI}
\end{algorithm}

\subsubsection{P-CSI}
\label{sec:PCSI_ii_iv}

In the case of P-CSI, as stated in the previous sub-section and discussed in Section~\ref{sec:SystemModel}, B is aware of the instantaneous channel realizations in the cellular links, but is only aware of the mean of the channel realization $\left|h_{Ui,Mi}\right|^2$.
In the case of the underlaying scheme (U.2), as stated in Theorem~\ref{thm:MaximumRU1}, to met the target outage probability it is only necessary to find the maximum $P_{Uj}$ allowed.
This can be done by evaluating the outage probability given by the left term of (\ref{eq:Pout_ii_iv}).
\begin{equation}
	P \left[ \C{\SNIR{\gamma_{Ui,Mi}}{\gamma_{Uj,Mi}}} < R_M \right] \leq P_{Out}.
	\label{eq:Pout_ii_iv}
\end{equation}
which can be accomplished through the methodology exposed in~\cite{5714150}.
In Algorithm~\ref{alg:MaximumRU1ComputationPCSI}, is shown the search approach on the domain of $P_{Uj}$ based on the bisection method,
\begin{algorithm}[t]
 \DontPrintSemicolon
 \SetAlgoLined
 \LinesNumbered

	\textbf{Input Parameters:}
	$P_{Uj}^{Max}$, $R_M$, $\left|h_{Uj,B}\right|^2$, $\left|h_{Ui,B}\right|^2$, $E\left[\left|h_{Ui,M}\right|^2\right]$, $E\left[\left|h_{Uj,M}\right|^2\right]$, $\epsilon$\;
	
	\textbf{Ouput Parameters:}
	$P_{Uj}$,$R_{Uj}$\;
	
	\textbf{Algorithm:}\;
	
	$P_{Uj}^{min} = 0$\;
	$P_{Uj}^{max} = P_{U}^{Max}$\;
	
	$P_{Uj}^{Current} = \frac{\left(P_{Uj}^{min} + P_{Uj}^{max}\right)}{2}$\;
	
	\While{$\left| P_{Out} - P_{Out}^{*}\left( P_{Uj}^{Current} \right) \right| > \epsilon$}{
		
			$P_{Uj}^{Current} = \frac{\left(P_{U1}^{min} + P_{Uj}^{max}\right)}{2}$\;
			
			\eIf{$P_{Out}^{*}\left( P_{Uj}^{Current} \right) > P_{Out}$}{
				$P_{Uj}^{max} = P_{Uj}^{Current}$\;
			}{
				$P_{Uj}^{min} = P_{Uj}^{Current}$\;
			}
	}
	
	$\gamma_{Uj,B} = P_{Uj}^{Current} \left|h_{Uj,B}\right|^2$\;
	
	\eIf{$R_M \leq \C{\SNIR{\gamma_{Ui,B}}{\gamma_{Uj,B}}}$}{
		$R_{Uj} = \C{\gamma_{Uj,B}} $\;
	}{
		$R_{Uj} = \C{\SNIR{\gamma_{Uj,B}}{\gamma_{Ui,B}}}$
	}
	\caption{Procedure to compute $P_{Uj}$ and $R_{Uj}$ in the case of P-CSI, given an $P_{Out}$ using a bi-section search method with stopping criterium $\epsilon$ where $P_{Out}^{*}\left(\dot{P}_{Uj}\right)$ is computed following the method exposed in~\cite{5714150}.}
\label{alg:MaximumRU1ComputationPCSI}
\end{algorithm}

\section{Numerical Results}
\label{sec:NumericalResults}

In this section the performance of the proposed underlaying and respective reference schemes are compared through numerical results.

\subsection{Scenario Settings and Performance Metrics}
\label{sec:ScenarioSettings}

The performance comparison is performed in two different network settings, $S_1$ and $S_2$, for which the relevant parameter settings are listed in Table~\ref{tab:SimulationScenarioSettings}.
\begin{table*}[t]
	\centering
		\begin{tabular}{ l c }
		\hline
		\textbf{Parameter} & \textbf{Value Range}\\ \hline
		$R_M$ & $0.5 [bps]$\\	
		$S_M$ & $10$ \\
		$P_{Out}$ & $0.1$\\
		$\epsilon$ & $0.01$\\
		$P_B^{max}$ & $46 [dBm]$~\cite{HarriHolma2011} \\
		$P_U^{max}$ & $24 [dBm]$~\cite{HarriHolma2011} \\
		$d_{max}$  & $500 [m]$ \\
		$d$ & $S_1 : d_{max}/2 [m]$ and $S_2 : \left[10, d_{max}/2\right] [m]$ \\
		$\left| h_{M_i,U_i} \right|^2$  & $-60 [dB]$~\cite{Yazdandoost2009,6227423}\\		
		$\left| h_{U,B} \right|^2$, $\left| h_{U_j,M_i} \right|^2$ & $-\left(35.74 \log_{10}\left(d\right) + 30.94\right) [dB]$~\cite[C1,NLOS,$f_c=2.1 GHz$]{WINNERII2007}\\
		$\sigma_{M}^2$	& $-97.5 [dBm]$~\cite{HarriHolma2011}\\
		$\sigma_{U}^2$  & $-97.5 [dBm]$~\cite{HarriHolma2011}\\
		$\sigma_{B}^2$  & $-116.5 [dBm]$~\cite{HarriHolma2011}\\
		\hline
		\end{tabular}
	\caption{Simulation scenario settings.}
	\label{tab:SimulationScenarioSettings}
\end{table*}

The first network setting, denoted by $S_1$, is the one of the network depicted in Figure~\ref{fig:SystemModel_and_SchedulingOptions}. While, the second network setting, denoted by $S_2$, shows the performance of the proposed underlaying schemes in the presence of spatial diversity by considering a network composed by four Ui-Mi pairs, deployed uniformly across the cell, and one B.

In the case of $S_2$, there is the need to define a scheduler that selects the $(i,j)$ pair to serve at each scheduling slot, e.g. for the (U.1) scheme the $(i,j)$ pair denotes the pairs Mi-Ui and Mj-Uj.
This selection is based on the maximization of the scheduler's scheduling metric.
The schedulers here considered are the MaxR and Proportional Fair (PF)~\cite{Pratas2007}, for which the scheduling metrics for underlaying schemes (U.1) and (U.2) are listed in Table~\ref{tab:SchedulingAlgorithms}.
The MaxR scheduler selects always the pair $(i,j)$ that maximizes the rate of the link associated with $j$ at a given scheduling slot; while the PF selects the pair $(i,j)$ so that the average rate over time $R_{B,U_j}^\ast$ for (U.1) and $R_{U_j,B}^\ast$ for (U.2) all possible $j$ links converges to the network mean rate.
\begin{table*}[t]
	\centering
		\begin{tabular}{ l c c}
		\hline
		\textbf{Scheduling Metric} & (U.1) & (U.2) \\ \hline
		MaxR & $R_{B,U_j}$ & $R_{U_j,B}$ \\	
		Proportional Fair & $R_{B,U_j}/R_{B,U_j}^\ast$ &  $R_{U_j,B}/R_{U_j,B}^\ast$\\
		\hline
		\end{tabular}
	\caption{Scheduling Algorithm Metrics.}
	\label{tab:SchedulingAlgorithms}
\end{table*}

The performance of the proposed underlaying schemes are compared with the respective reference schemes, using the following performance metrics.

The mean normalized rate, $E\left[\Gamma_x\right]$, which is the mean of the instantaneous rate normalized by the number of time slots required by each of the scheduling schemes in the direction $x$.
In the downlink is given by,
\begin{equation}
	E\left[\Gamma_D\right] = E\left[\frac{R_{B,Uj}}{N}\right]
\end{equation}
while in the uplink by,
\begin{equation}
	E\left[\Gamma_U\right] = E\left[\frac{R_{Uj,B}}{N}\right]
\end{equation}
where $N$ is the number of time slots in the scheduling frame, e.g. in (R.1) $N=2$, while in (U.1) $N=1$.

The mean active time per link gives the fraction of time that each link is active.
Note that in the reference schemes each link cannot be active more than $50\%$ of the time, while for the underlaying schemes there is no such limitation.

\subsection{Performance Evaluation of (U.1)}
\label{sec:PerformanceEvaluationOfi_iii}

The performance of the (U.1) scheme in the network setting $S_1$ is depicted in Figure~\ref{fig:Reference_And_Mi-Ui+B-Uj_Comparison}.
The underlaying scheme (U.1) is more effective in the presence of strong interferers, which explains why when B transmits at full power ($46 dBm$), the $E\left[\Gamma_D\right]$ achieved by the underlaying scheme for F-CSI and P-CSI is higher than the one for (R.1).
It should be noted that P-CSI has lower $\bar{\Gamma}_D$ than F-CSI, since it has less information on the channel state and therefore takes a more conservative value for $R_{B,U_2}$.
\begin{figure}
	\centering
		\includegraphics[width=\linewidth]{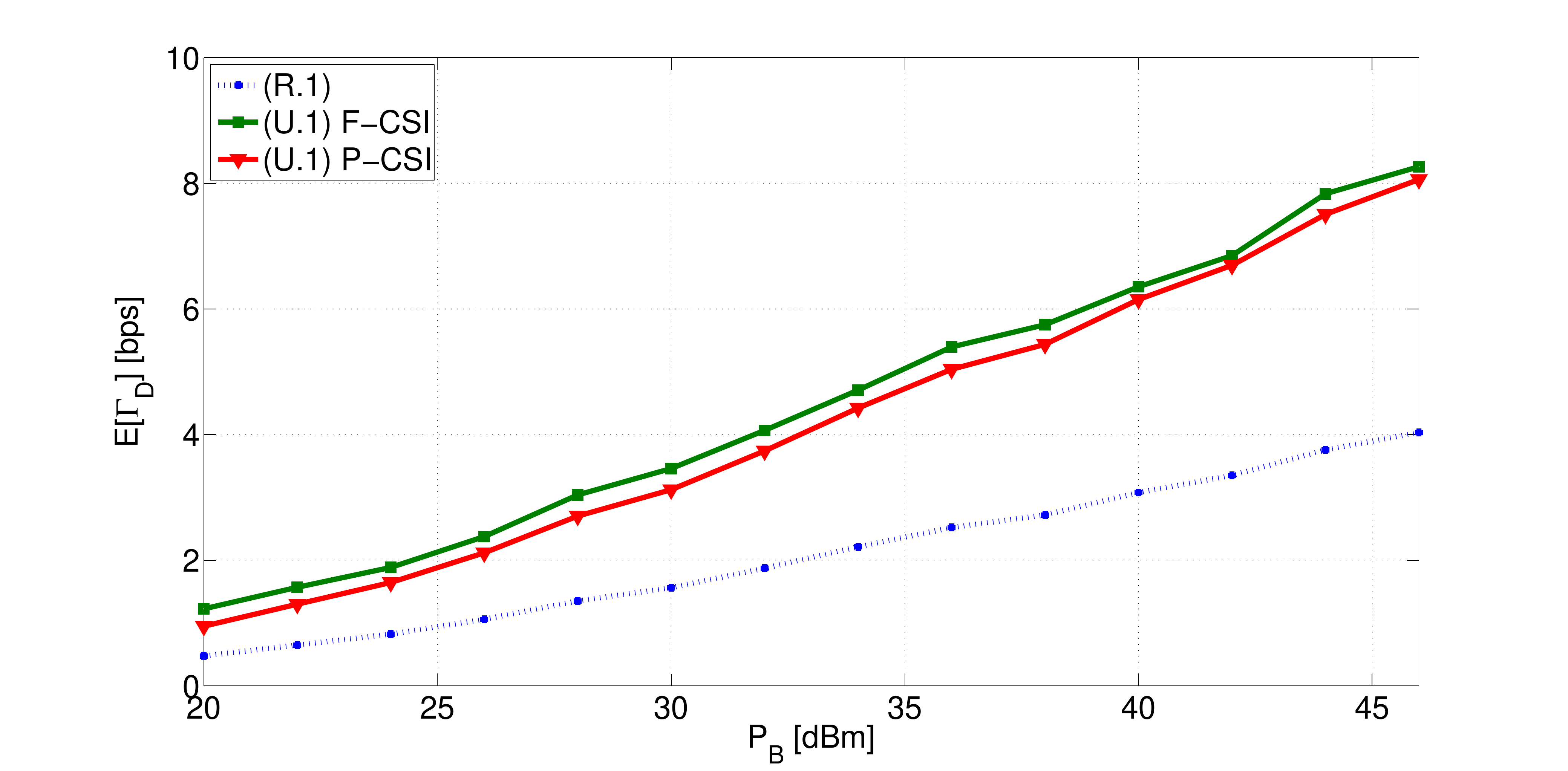}
	\caption{Comparison of (U.1) and (R.1) in $S_1$.}
	\label{fig:Reference_And_Mi-Ui+B-Uj_Comparison}
\end{figure}

In Figure~\ref{fig:i_iii_network_evaluation} are depicted the performance of (U.1) and (R.1) in the network setting $S_2$.
It can be observed, that (U.1) allows the links to be active more often and an higher $E\left[\Gamma_D\right]$ than (R.1), in both schedulers.
We note that the reason why communications occur more often is due to the underlaying of the information streams in the same time slot.
\begin{figure}
	\centering
		\includegraphics[width=\linewidth]{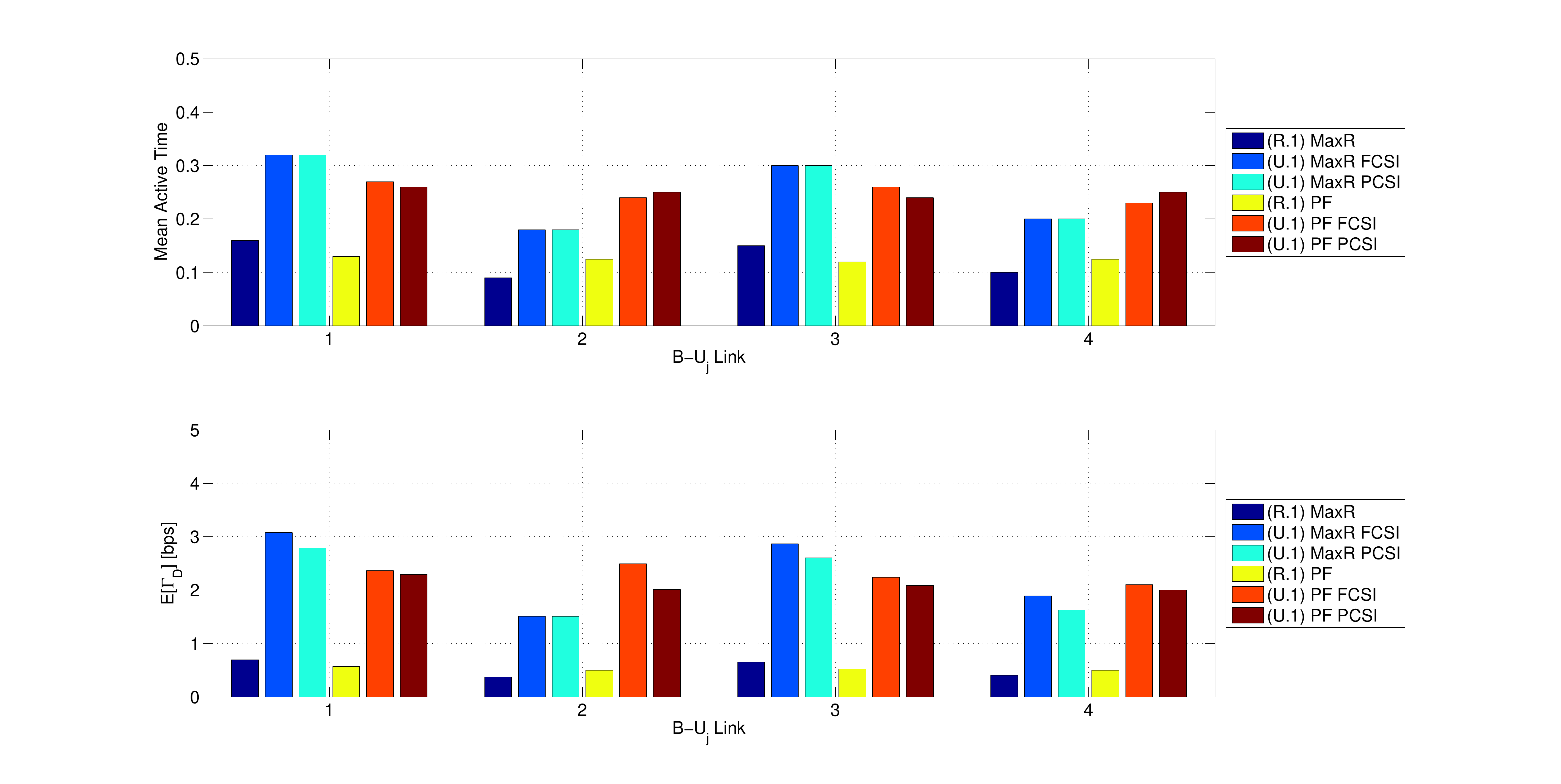}
	\caption{Comparison of (U.1) and (R.1) in $S_2$.}
	\label{fig:i_iii_network_evaluation}
\end{figure}

\subsection{Performance Evaluation of (U.2)}
\label{sec:PerformanceEvaluationOfIiIv}

In Figure~\ref{fig:Reference_And_Uj-B+Ui-Mi_Comparison} is depicted the comparison of the performance of (U.2) and (R.2) in the network setting $S_1$.
As expected, the achievable $E\left[\Gamma_U\right]$ for (U.2) is much higher than (R.2).
Although, as exposed in Section~\ref{sec:UnderlayedSchemesCharacterization}, this gain is only possible when Ui transmits with maximum power to Mi, so to maximize both the outage margin at Mi and to ensure that the transmission Ui-Mi can be successfully cancelled at B through OIC.
\begin{figure}
	\centering
		\includegraphics[width=\linewidth]{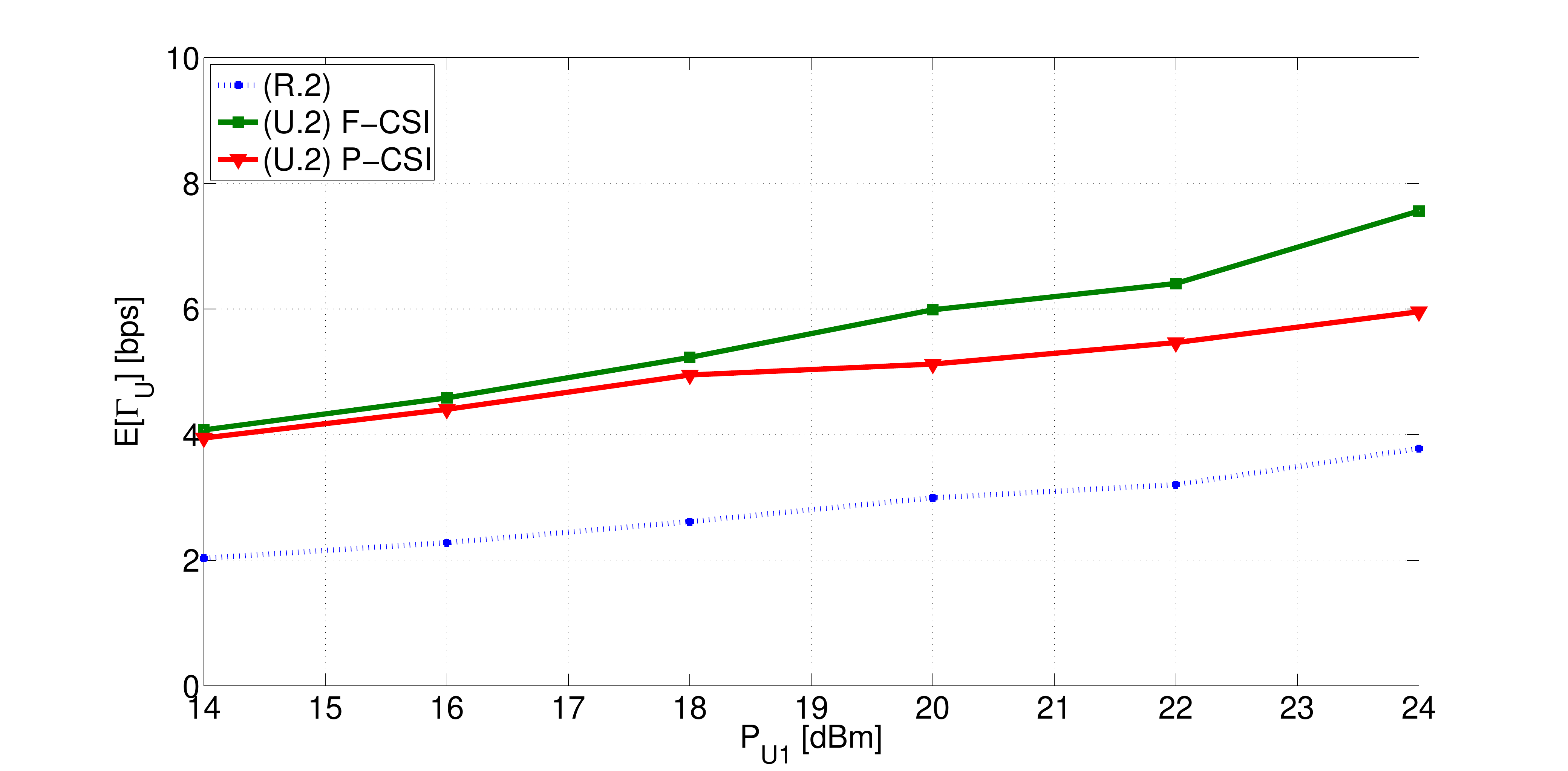}
	\caption{Comparison of (U.2) and (R.2) in $S_1$.}
	\label{fig:Reference_And_Uj-B+Ui-Mi_Comparison}
\end{figure}

In Figure~\ref{fig:ii_iv_network_evaluation} is depicted the comparison performance over successive scheduling slots of (U.2) and (R.2), with MaxCI and PF schedulers for F-CSI and P-CSI.
As expected there is a substantial gain both in mean link active time as well as in $E\left[\Gamma_U\right]$.
\begin{figure}
	\centering
		\includegraphics[width=\linewidth]{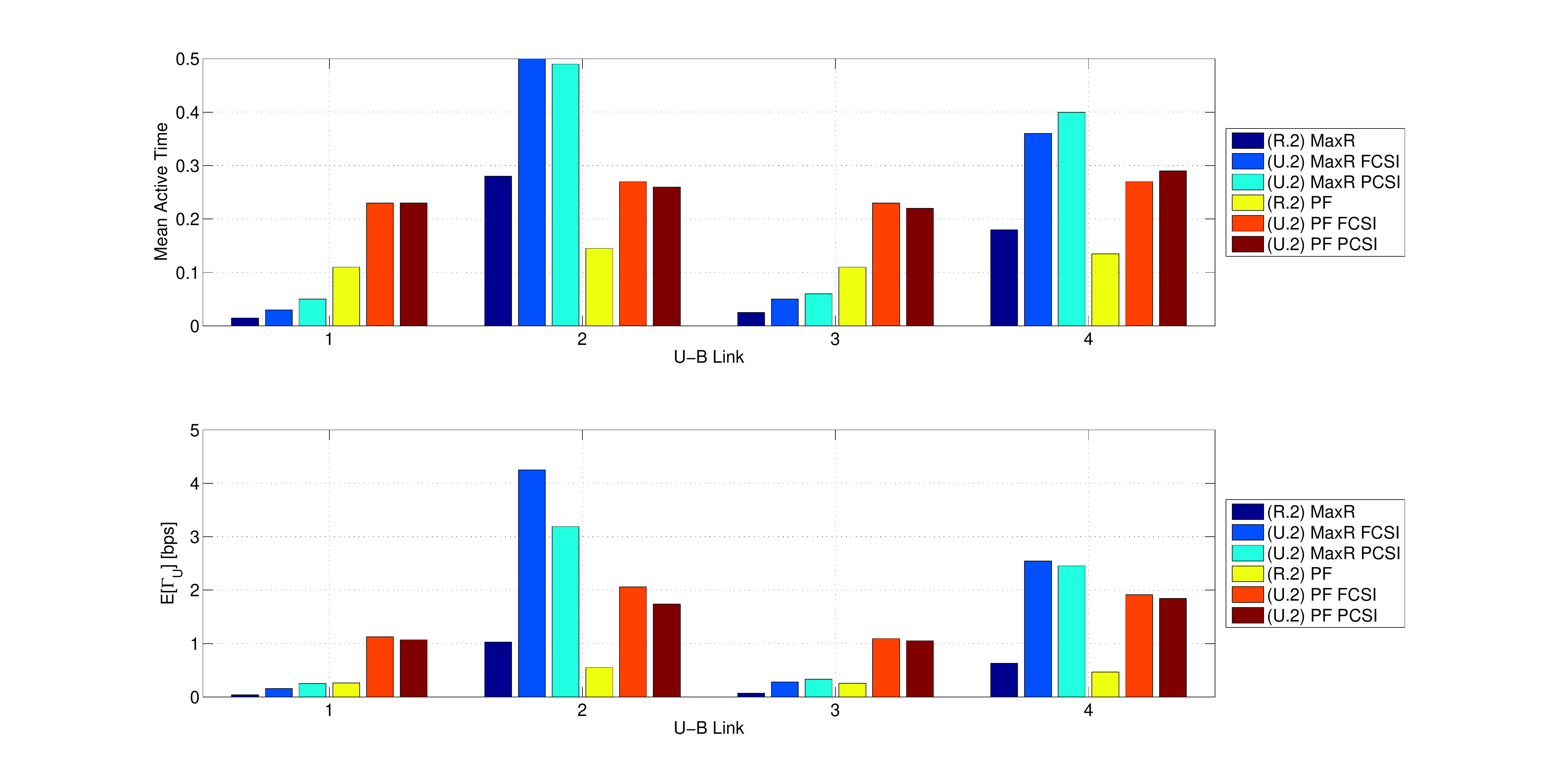}
	\caption{Comparison of (U.2) and (R.2) in $S_2$.}
	\label{fig:ii_iv_network_evaluation}
\end{figure}

\section{Conclusion}
\label{sec:Conclusion}

The current trend in cellular communications is to embrace the device-to-device (D2D) paradigm as a mean to off-load, from the cellular infrastructure, the transmissions between devices in close proximity.
Therefore, both industry and academia are focusing on how to enable these exchanges to occur directly between devices, but chiefly towards high-bit-rate applications.
There is also an increased interest on deploying Machine-Type devices within cellular networks.
These devices will share the same cellular air interface, but will be less capable than normal cellular devices.
Where this reduced capability refers to lower transmission power, less signal processing and a fixed (low) transmission rate.
The main consequence of these reduced capabilities will be the severe outage problems in particular to the devices at the cell edge.
The logical step is to allow these devices to cooperate with normal devices, making use of the device-to-device communication paradigm.

In this paper we proposed a network assisted device-to-device solution that enables the cooperation between Machine-Type cellular devices and normal cellular devices, with the goal of reducing the outage the machine-type communications while maximizing the combined cellular system rate. The model captures the main performance objectives: outage probability for a given low fixed rate for the Machine-Type communication versus rate maximization for the cellular broadband traffic. System rate maximization is accomplished by the underlaying of device-to-device and cellular links both in uplink and downlink directions.
The link underlaying is accomplished by applying the information-theoretic multiple access channel principles, opportunistic interference cancellation, as well as rate and power adaptation of the cellular links, in order to ensure that the target outage probability of machine type D2D connection is attained.We propose two underlaying schemes that act in the downlink and uplink directions of the cellular traffic, respectively. The performance benefits are demonstrated both through analysis and simulation results. 

The results shown in this paper motivate that low power, fixed rate Machine-Type devices can share the air interface with normal cellular devices.
This paves the way, that a single air interface technology can be used to connect several kinds of devices; and therefore future devices could converge to a single air interface technology instead of the current state of affairs where depending of the application there is a different applicable air interface technology.
\section*{Acknowledgment}

The research presented in this paper was partly supported by the Danish Council for Independent Research (Det Frie Forskningsr\aa d) within the Sapere Aude Research Leader program, Grant No. 11-105159 ``Dependable Wireless Bits for Machine-to-Machine (M2M) Communications'' and performed partly in the framework of the FP7 project ICT-317669 METIS, which is partly funded by the European Union

\appendices

\section{Proof of Lemma~\ref{lem:OutageMargin}}
\label{sec:ProofOfLemmaRef}

	The probability of outage, $P_{out}$, in the absence of interference, is given as,
	\begin{equation}
		P_{Out} = P \left[\gamma_{X,Y} < \gamma_{M}\right] = 1 - e^{-\frac{\gamma_{M}}{\bar{\gamma}_{X,Y}}}
	\end{equation}
	which can be manipulated to held,
	\begin{equation}
		\gamma_{M} = - \bar{\gamma}_{X,Y} \ln\left(1 - P_{Out}\right).
		\label{eq:gammaM}
	\end{equation}

	The mean SNR, $\bar{\gamma}_{X,Y}$, is given by,
	\begin{equation}
		\bar{\gamma}_{X,Y} = \frac{P_X \left| \bar{h}_{X,Y} \right|}{\sigma_Y^2}
		\label{eq:meangamma}
	\end{equation}
	where $P_X$ is the transmission power of node $X$.
	
  The average outage rate~\cite{Choudhury2007}, $R_M^{Out}$, is defined as,
	\begin{equation}
		R_M^{Out} = \left( 1 - P_{Out} \right) \log \left( 1 + \gamma_{M} \right)
		\label{outageRate}
	\end{equation}
	which can be manipulated using (\ref{eq:gammaM}) and (\ref{eq:meangamma}) to held,
	\begin{equation}
		P_X^{Min} = \frac{\sigma_Y^2}{\left|h_{X,Y}\right|^2} \cdot \frac{2^{\frac{R_M^{Out}}{1 - P_{Out}}}-1}{\ln \left( 1 - P_{Out} \right)}
		\label{eq:lowerBoundPi}
	\end{equation}
	giving the lower bound of $P_i$ so that $R_M^{Out}$ and $P_{Out}$ are met.
	
	Introducing interferer $k$, then for link X-Y to maintain the $R_M$ and $P_{Out}$ requirements, then the following relation must hold,
	\begin{equation}
		\frac{P_X^{Min}\left| h_{X,Y} \right|^2}{\sigma_Y^2} \leq \frac{P_X\frac{\left|h_{X,Y}\right|^2}{\sigma^2}}{1 + \frac{1}{\sigma^2} \sum_k^N P_k \left| h_{k,Y} \right|^2}
	\end{equation}
	which can be manipulated to held,
	\begin{equation}
		P_X \geq \underbrace{\left(1 + \frac{1}{\sigma_Y^2} \sum_k^N P_k \left| h_{k,Y} \right|^2\right)}_{S_M} P_X^{Min}
	\end{equation}
	%


\end{document}